\documentclass{llncs}

\usepackage{amssymb}

\usepackage{amsfonts}
\usepackage{amsmath}
\usepackage{breqn}
\usepackage{verbatim}
\usepackage{bbm}
\usepackage{ifthen}
\usepackage{pbox}
\usepackage{todonotes}

\usepackage[linesnumbered, boxed, ruled]{algorithm2e}
\usepackage{algorithmic}

\usepackage{mathtools}

\usepackage[ruled]{algorithm2e}

\SetAlFnt{\small}
\SetAlCapFnt{\small}
\SetAlCapNameFnt{\small}
\SetAlCapHSkip{0pt}
\IncMargin{-\parindent}

\title{On Welfare Approximation and Stable Pricing}
\author{
Michal Feldman\inst{1}
\and
Nick Gravin\inst{2}
\and
Brendan Lucier\inst{2}
}

\institute{Tel-Aviv University; \email{mfeldman@tau.ac.il}
\and Microsoft Research; \email{ngravin@gmail.com}, \email{brlucier@microsoft.com}
}

\begin{document}


\newcommand{\AutoAdjust}[3]{\mathchoice{ \left #1 #2  \right #3}{#1 #2 #3}{#1 #2 #3}{#1 #2 #3} }
\newcommand{\Xcomment}[1]{{}}
\newcommand{\eval}[1]{\left.#1\vphantom{\big|}\right|}
\newcommand{\inteval}[1]{\Big[#1\Big]}
\newcommand{\InParentheses}[1]{\AutoAdjust{(}{#1}{)}}
\newcommand{\InBrackets}[1]{\AutoAdjust{[}{#1}{]}}
\newcommand{\Ex}[2][]{\operatorname{\mathbf E}_{#1}\InBrackets{#2\vphantom{E_{F}}}}
\newcommand{\Exlong}[2][]{\operatornamewithlimits{\mathbf E}\limits_{#1}\InBrackets{#2\vphantom{\operatornamewithlimits{\mathbf E}\limits_{#1}}}}
\newcommand{\Prx}[2][]{\operatorname{\mathbf{Pr}}_{#1}\InBrackets{#2}}
\newcommand{\Prlong}[2][]{\operatornamewithlimits{\mathbf Pr}\limits_{#1}\InBrackets{#2\vphantom{\operatornamewithlimits{\mathbf Pr}\limits_{#1}}}}
\def\prob{\Prx}
\def\expect{\Ex}
\newcommand{\super}[1]{^{(#1)}}
\newcommand{\dd}{\mathrm{d}}  
\newcommand{\given}{\;\mid\;}

\newcommand{\be}{\begin{equation}}
\newcommand{\ee}{\end{equation}}
\newcommand{\argmin}{\mathop{\rm argmin}}
\newcommand{\argmax}{\mathop{\rm argmax}}
\newcommand{\vr}[1]{{\mathbf{#1}}}
\newcommand{\bydef}{\stackrel{\bigtriangleup}{=}}
\newcommand{\eps}{\varepsilon}
\newcommand{\mm}[1]{\mathrm{#1}}
\newcommand{\mc}[1]{\mathcal{#1}}
\newcommand{\mb}[1]{\mathbf{#1}}
\newcommand{\vect}[1]{\ensuremath{\mathbf{#1}}}
\newcommand{\R}{\mathbb{R}}

\def \EE   {{\mathbb E}}
\def \OPT {\mathcal{OPT}}
\def \vf  {\textrm{vf}}
\def \dvf {\varphi}
\def \utility {u}
\def \reals {{\mathbb R}}

\newcommand{\Idr}[1]{\mathds{1}\InBrackets{#1\vphantom{\sum}}}


\newcommand{\dist}{\mathcal{F}}
\newcommand{\disti}[1][i]{{\mathcal{F}_{#1}}}
\newcommand{\distsmi}[1][i]{\dists_{\text{-}#1}}
\newcommand{\dists}{\vect{\dist}}
\newcommand{\distw}{\widetilde{\mathcal{F}}}

\newcommand{\valdist}{\mathcal{F}}
\newcommand{\valdists}{\vect{\valdist}}
\newcommand{\valdisti}[1][i]{{\valdist_{#1}}}
\newcommand{\valdistsmi}[1][i]{\valdists_{\text{-}#1}}

\newcommand{\dens}{f}
\newcommand{\denss}{\vect{ \dens}}
\newcommand{\densi}[1][i]{{\dens_{#1}}}

\newcommand{\agents}{N}
\newcommand{\nagent}{n}
\newcommand{\items}{M}
\newcommand{\nitem}{m}
\newcommand{\auction}{A}

\newcommand{\weight}{w}
\newcommand{\partition}{\Gamma}
\newcommand{\partitioni}{\Gamma_i}

\newcommand{\CWE}[0]{\textsf{CWE}}
\newcommand{\ef}[0]{envy free}
\newcommand{\EF}[0]{EF}
\newcommand{\CWEWOMC}[0]{CWE without market clearance}
\newcommand{\SW}[0]{\textsf{SW}}
\newcommand{\Revenue}[0]{\textsf{Rev}}

\newcommand{\bid}{b}
\newcommand{\bids}{\vect{\bid}}
\newcommand{\bidsmi}[1][i]{\bids_{\text{-}#1}}
\newcommand{\bidi}[1][i]{{\bid_{#1}}}

\newcommand{\val}{v}
\newcommand{\vals}{\vect{\val}}
\newcommand{\valsmi}[1][i]{\vals_{\text{-}#1}}
\newcommand{\vali}[1][i]{{\val_{#1}}}
\newcommand{\valith}[1][i]{{\val_{(#1)}}}

\newcommand{\wal}{\widetilde{v}}
\newcommand{\wals}{\vect{\wal}}
\newcommand{\walsmi}[1][i]{\wals_{\text{-}#1}}
\newcommand{\wali}[1][i]{{\wal_{#1}}}

\newcommand{\util}{u}
\newcommand{\utils}{\vect{\util}}
\newcommand{\utili}[1][i]{\util_{#1}}
\newcommand{\utilsmi}[1]{\utils_{\text{-}#1}}

\newcommand{\price}{p}
\newcommand{\prices}{\vect{\price}}
\newcommand{\pricei}[1][i]{{\price_{#1}}}

\newcommand{\type}{t}
\newcommand{\types}{\vect{\type}}
\newcommand{\typei}[1][i]{{\type_{#1}}}
\newcommand{\typesmi}[1][i]{\types_{\text{-}#1}}

\newcommand{\alloc}{X}
\newcommand{\allocs}{\vect{\alloc}}
\newcommand{\allocsmi}[1][i]{\allocs_{\text{-}#1}}
\newcommand{\alloci}[1][i]{{\alloc_{#1}}}

\newcommand{\opt}{\text{OPT}}
\newcommand{\opts}{\vect{ \opt}}
\newcommand{\opti}[1][i]{{\opt_{#1}}}
\newcommand{\optsmi}[1][i]{\opts_{\text{-}#1}}

\newcommand{\talloc}{Y}
\newcommand{\tallocs}{\vect{\talloc}}
\newcommand{\talloci}[1][i]{{\talloc_{#1}}}
\newcommand{\tallocsmi}[1][i]{\tallocs_{\text{-}#1}}

\newcommand{\rank}{r}
\newcommand{\ranks}{\vect{\rank}}
\newcommand{\ranki}[1][i]{{\rank_{#1}}}
\newcommand{\ranksmi}[1][i]{\ranks_{\text{-}#1}}

\newcommand{\crit}{\theta}
\newcommand{\crits}{\vect{\crit}}
\newcommand{\criti}[1][i]{{\crit_{#1}}}
\newcommand{\critsmi}[1][i]{\crits_{\text{-}#1}}

\newcommand{\decl}{d}
\newcommand{\decls}{\vect{\decl}}
\newcommand{\decli}[1][i]{{\decl_{#1}}}
\newcommand{\declsmi}[1][i]{\decls_{\text{-}#1}}

\newcommand{\demand}{D}
\newcommand{\demands}{\vect{\demand}}
\newcommand{\demandi}[1][i]{{\demand_{#1}}}
\newcommand{\demandsmi}[1][i]{\demands_{\text{-}#1}}

\newcommand{\sold}{{\texttt{SOLD}}}
\newcommand{\soldi}[1][i]{\sold_{#1}}

\newcommand{\CWEalg}{{\sc CWE} algorithm}

\newcommand{\Pool}{\text{Pool}}
\newcommand{\Pop}{\textbf{Pop}}
\newcommand{\Push}{\textbf{Push}}
\newcommand{\Bundle}{\textbf{Bundle}}
\newcommand{\ResolveConflict}{\textbf{ResolveConflict}}


\newcommand{\Reject}{\text{Reject}}
\newcommand{\AllocateDemand}{\textbf{AllocateDemand}}
\newcommand{\RaisePrices}{\textbf{RaisePrices}}

\maketitle

\begin{abstract}
We study the power of item-pricing as a tool for approximately optimizing social welfare in a combinatorial market.
We consider markets with $m$ indivisible items and $n$ buyers.
The goal is to set prices to the items so that, when agents purchase their most demanded sets simultaneously, no conflicts arise and the obtained allocation has nearly optimal welfare.
For gross substitutes valuations, it is well known that it is possible to achieve optimal welfare in this manner.
We ask: can one achieve approximately efficient outcomes for valuations beyond gross substitutes?
We show that even for submodular valuations, and even with only two buyers, one cannot guarantee an approximation better than $\Omega(\sqrt{m})$.
The same lower bound holds for the class of single-minded buyers as well.
Beyond the negative results on welfare approximation, our results have daunting implications on {\em revenue} approximation for these valuation classes: in order to obtain good approximation to the collected revenue, one would necessarily need to abandon the common approach of comparing the revenue to the optimal welfare; a fundamentally new approach would be required.

\end{abstract}

\section{Introduction}
Combinatorial markets are a canonical setting in which economic and computational theory collide.  In a combinatorial market scenario, a variety of indivisible and heterogeneous goods are to be allocated among buyers with idiosyncratic and potentially complex preferences over subsets.  A natural economic goal is to find (and implement) allocations that are efficient, meaning that the total value generated by the allocation is maximized.  Over the past decade, the algorithmic mechanism design community has made a concerted effort to understand the interplay between the requirement of aligning buyer incentives (from economics) and the tool of approximation (from computer science) in the context of this welfare maximization problem.


An important observation from economic theory is that, in many scenarios, markets can be resolved using a very simple instrument: item prices.  Rather than execute an auction or other mechanism,
one sets an appropriate price for each good and allow buyers to select their favorite bundles.  Indeed, this is the strategy employed by most retail outlets in the western world, even in scenarios where items are truly heterogeneous (e.g., antique stores).  This prevelance can be attributed to the natural and distributed manner in which buyers can respond to prices posted by a seller.


In markets with limited supply, a desirable property of item prices is that \emph{each} buyer can \emph{simultaneously} obtain her most demanded set.  That is, no item is overdemanded at the specified prices.  This is a natural market fairness criterion, and (as argued elsewhere \cite{Guruswami2005}) is important for maintaining buyer satisfaction. Following \cite{FGL-13}, we will say that a choice of prices is \emph{stable} if it satisfies this fairness condition.
When each buyer wants at most a single item, this property is analogous to envy-freeness (no buyer prefers another's outcome to her own).  Hence, for unit-demand buyers, stable prices are equivalent to envy-free prices \cite{Guruswami2005}.

It is well known that when each buyer wants at most a single item (i.e., unit-demand preferences), or when the value a buyer derives from an item is independent of what other items she is allocated (i.e., additive preferences), there always exist stable prices
for which the resulting market outcome achieves optimal efficiency \cite{Shapley1971}. This is referred to as a Walrasian (or competitive) equilibrium.  Such prices exist more generally for the class of gross substitutes preferences \cite{Kelso1982}. This class includes all additive and unit-demand valuations, but is strictly contained in the class of submodular valuations.  To summarize, when buyers have gross substitutes valuations, item pricing can efficiently resolve a market.

In this work we study the power of stable item prices to approximate welfare when valuations do not satisfy the gross substitutes condition. It is known that for any class of valuations beyond gross substitutes, there exist instances for which no stable prices can clear the market in an optimally efficient assignment \cite{Gul1999}.
We ask: for what classes of buyer preferences do there exist prices that generate \emph{approximately} efficient outcomes, when buyers select their utility-maximizing sets?  Is it possible to recover a constant fraction of the optimal welfare when buyer valuations are, for example, submodular?

\subsection{Our results}

It might not come at a surprise that item prices are not sufficient to obtain high welfare if items are complementary.
After all, if items have value only when sold together, then it seems intuitive to bundle (rather than setting individual prices).
Perhaps more surprisingly, we also show that a polynomial gap may exist even if all valuations are submodular. This is despite the 
existential nature of the question, which carries the assumption that buyer types are publicly known.  Even for two buyers with 
submodular valuations, there may not exist prices that attain any reasonable approximation to the optimal welfare.  Crucially, 
these lower bounds do not stem from computational considerations, but rather are existential.  That is, stable item prices 
necessarily provide a poor approximation for welfare.




\vspace{0.1in}\noindent {\bf Main Result:}
{\em There exists an instance with two buyers with submodular valuations, such that there does not exist a stable outcome that obtains better than an $\Omega(\sqrt{m})$-approximation to the optimal social welfare.}

\vspace{0.1in} \noindent

The result is formulated for the objective of welfare maximization. A different, and also very natural, objective is seller revenue, which has been 
extensively studied in the literature on the design of revenue-maximizing envy-free prices \cite{Balcan2008,Briest06,Cheung2008,Hartline2011}. 
Although our results do not have a direct implication for the revenue objective, as far as we are aware, all existing methods for approximating envy-free revenue proceed by comparing to the benchmark of optimal welfare. An implication of our result is that such an approach cannot generalize to submodular valuations and beyond.  Since the welfare generated via stable prices cannot approximate optimal welfare, certainly the revenue from stable prices cannot approximate this benchmark either. Any approximation analysis for the revenue of stable prices for submodular valuations would therefore require a fundamentally new approach.



\vspace{0.1in}
These negative results may lead one to suspect that a super-constant lower bound is unavoidable for any valuation class beyond gross substitutes, even for two buyers (similar in spirit to the result of \cite{Gul1999}, showing that the class of gross substitutes valuations is maximal with respect to Walrasian equilibrium existence). In Appendix \ref{app:budget-additive} we show that this is not the case: for two budget-additive buyers, there always exists a set of item prices that achieves a quarter of the optimal welfare.

We also consider single-minded valuations, for which we give a lower bound of $\Omega(\sqrt{m})$.
For this case, a matching upper bound of $O(\sqrt{m})$ follows from the results in \cite{Cheung2008}.
These results are obtained via an analysis of the configuration LP and its dual.

\vspace{0.1in}\noindent {\it Techniques.}
Our lower bound construction is based upon a connection between the existence of Walrasian prices and the integrality gap of the configuration LP for a market, but requires new ideas. In particular, a simple configuration LP argument presented in \cite{FGL-13} showing a lower bound of $\Omega(m)$ for the broader class of fractionally subadditive valuations completely fails for the class of submodular valuations. One reason is that the example in \cite{FGL-13} works by exploiting the fact that complement-free valuations can have ``hidden complementarities:" a marginal valuation function\footnote{Formally, given a valuation $\val$ on a set of items $M$ and a set $S \subseteq M$, the marginal valuation of a set $T \subseteq M \setminus S$ is defined by $\val(T|S) = \val(T \cup S) - \val(S)$.} of a complement-free valuation (i.e., the extra value derived from additional items on top of a fixed set of items) is not necessarily complement-free. One might therefore guess that submodular valuations, which do not exhibit such hidden complementarities~\cite{Lehmann2001}, might overcome this lower bound and yield a constant approximation. Another reason is that the example from \cite{FGL-13} uses symmetric valuation functions, but symmetric submodular functions are gross substitutes~\cite{Gul1999}, and thus lower bounds based on symmetric constructions cannot apply.

\subsection{Related work}


Our work is related to the literature on {\em algorithmic pricing} \cite{Guruswami2005}, which studies the problem of setting envy-free item prices in full-information settings.  This work focuses mainly on algorithms for finding prices that approximately maximize the seller's revenue, assuming that buyers are unit-demand\footnote{Note that, for unit-demand buyers, envy-freeness is equivalent to our notion of stability.}  \cite{Balcan2008,Briest06,Cheung2008,Hartline2011}.


One motivation for our work is the non-existence of Walrasian equilibrium in general combinatorial markets.
Characterizations of existence of Walrasian equilibria were studied in, for example, \cite{Aumann1975,Kelso1982,Leonard1983,Bikhchandani1997,Gul1999}.
Extensions of envy-free pricing for valuation classes beyond gross substitutes were studied in \cite{MuAlem2009,Fiat2009}, but whereas we focus on item prices those works consider more general payment rules that assign prices to bundles of goods.

Our model can be viewed as a relaxation of the notion of Walrasian equilibrium, where we do not insist that all items are sold.
A similar relaxation was recently considered in \cite{FGL-13}, showing that if, in addition to our relaxation, one can also bundle items into packages prior to setting prices, then at least half of the optimal social welfare can be obtained for arbitrary valuations.
Our negative result provides additional motivation for
such a bundling operation.
Indeed, our results illustrate that insisting on outcomes that simultaneously satisfy the demands of all agents, using item prices, can prevent a good approximation to social welfare.
The results of \cite{FGL-13} show that this difficulty can also be circumvented via bundling.

A different relaxation of Walrasian equilibrium was presented in \cite{FKL-12}.  In their equilibrium notion, no buyer wishes to add additional items to his allocation (but might prefer a non-superset of their current allocation).  This can be viewed as a partial relaxation of the buyer-side equilibrium constraints in a Walrasian equilibrium.  In contrast, our approach is to relax the seller-side constraint that all items must be sold.

\section{Model and Preliminaries}
\label{sec:model}
The setting considered in this work consists of a set $\items$ of $\nitem$ indivisible objects and a set of $n$ buyers.
Each buyer has a valuation function $\vali(\cdot) : 2^\items \to \reals_{\geq 0}$ that indicates his value for
every set of objects. 
As standard, we assume that valuations are monotone non-decreasing (i.e., $\vali(S) \leq \vali(T)$ for every $S \subseteq T \subseteq \items$) and are normalized so that $\vali(\emptyset) = 0$.
The profile of buyer valuations is denoted by $\vals=(\val_1,\dotsc,\val_n)$.

An {\em allocation} of $M$ is a vector of sets $\allocs = (\alloc_0, \alloc_1, \dotsc, \alloc_n)$, where $\alloci$ denotes the bundle assigned to buyer $i$, for $i \in [n]$, and $\alloc_0$ is the set of unallocated objects; i.e., $\alloc_0 = M \setminus \cup_{i \in [n]}\alloci$. It is required that $\alloc_i \cap \alloc_k = \emptyset$ for every $i\neq k$.
The social welfare of an allocation $\allocs$ is $\SW(\allocs)=\sum_{i=1}^{n}\vali(\alloci)$, and the optimal welfare is denoted by OPT.
An allocation $\allocs$ gives an $\alpha$-approximation for the social welfare if $\SW(\allocs) \geq (1/\alpha) \cdot \opt$.
A price vector $\prices=(\price_1, \dotsc, \price_{m})$ consists of a price $\price_j$ for each object $j \in M$.
As standard, we assume that each buyer has a quasi-linear utility function; i.e., the utility of buyer $i$ being allocated bundle $\alloc_i$ under prices $\prices$ is $\utili(\alloci, \prices) = \vali(\alloci) - \sum_{j \in \alloci}\price_j.$

Given prices $\prices=(\price_1, \dotsc, \price_{m})$, the {\em demand correspondence} $\demandi(\prices)$ of buyer $i$ contains the sets of objects that maximize buyer $i$'s utility:
\[
\demandi(\prices) = \left\{S^*:  S^* \in\argmax_{S\subseteq M}\{\utili(S,\prices)\}\right\}.
\]
An allocation $\alloci$ is optimal for buyer $i$ with respect to prices $\prices$ if $\alloci \in \demandi(\prices)$.
A tuple $(\allocs,\prices)$ is said to be {\em stable} if for every buyer $i$, $\alloci$ is optimal for $i$ with respect to $\prices$.  We will sometimes refer to such a tuple as a \emph{stable outcome}.  We refer to a pricing $\prices$ as being {\it stable} if there exists a supporting allocation $\allocs$ such that $(\allocs,\prices)$ is stable.




\subsection{Characterization}
\label{sec:config-lp}

In this section, we present a characterization of stable outcomes.
This characterization makes use of the configuration linear program (LP), which encodes the problem of maximizing social welfare over all fractional allocations of a combinatorial market.
The configuration LP is given by the following linear program:

\begin{eqnarray}
\max & & \sum_{i,S} \vali(S)\cdot x_{_{i,S}} \nonumber\\
\mbox{s.t.} & & \sum_{S \subseteq M} x_{_{i,S}} \leq 1 \mbox{ for every } i \in N \nonumber\\
& & \sum_{i, S \ni j} x_{_{i,S}} \leq 1 \mbox{ for every } j \in M \nonumber\\
& & x_{_{i,S}} \in [0,1] \mbox{ for every } i \in N, S \subseteq M 
\label{eq:config-lp}
\end{eqnarray}

It was shown in \cite{BM97} that a Walrasian Equilibrium exists if and only if the integrality gap of the configuration LP is $1$.  Since, by definition, an outcome is stable if and only if it is a WE over the set of items that it allocates, A straightforward corollary of that result is as follows:

\begin{corollary}
\label{cor:config.LP}
For any $M' \subseteq M$, there exists a stable tuple $(\allocs, \prices)$ with $\bigcup_{i} \alloci = M'$ iff the configuration LP restricted to the items in $M'$ has integrality gap $1$.
\end{corollary}

\subsection{Valuation Classes}
The following hierarchy of complement-free valuation classes have been studied extensively in the literature.
\begin{itemize}
\item $\val$ is \emph{additive} if $\val(S) = \sum_{j \in S}\val(\{j\})$ for all $S \subset \items$. 
%
%
\item $\val$ is \emph{XOS} if there exists a collection of additive functions
$A_1(\cdot),\ldots ,A_k(\cdot)$ 
such that for every set $S \subseteq \items$,  $\val(S) = \max_{1\le i\le k}A_i(S)$.
\item $\val$ is \emph{submodular} if for every two sets $S \subseteq T \subseteq M$ and item $j \in M$, $\val(j|T) \leq \val(j|S)$ (where $\val(j|S) \coloneqq \val(S \cup \{j\}) - \val(S)$ for every set $S \subseteq M$).
\item $\val$ is \emph{budget additive} if there is a value $B \geq 0$:
$\val(S) = \min\{ \sum_{j \in S}\val(\{j\}), B \}$ for every set $S \subseteq M$.
\end{itemize}


These valuations exhibit the following hierarchy: additive $\subset$ budget additive $\subset$ submodular $\subset$ XOS,
%
%
where all containments are strict.  
See \cite{Lehmann2001} for a detailed discussion.

\section{Submodular Valuations}
\label{sec:submodular}

The following theorem, which is the main result of the paper, shows that even for two submodular buyers, a large gap in welfare might be unavoidable in stable outcomes.

%

\begin{theorem}
\label{thm:submod.sim}
There exists an instance with two buyers with submodular valuations, such that there does not exist a stable outcome that obtains better than $\Omega(\sqrt{m})$-approximation to the optimal social welfare.
\end{theorem}

\begin{proof}
Before going into the details of the proofs, let us provide a high-level description of our approach.  By Corollary \ref{cor:config.LP}, we wish to construct a pair of submodular valuations for which the configuration LP has integrality gap strictly greater than $1$, even as items are removed from the market, unless almost all items are removed.  We begin by considering an instance that is not submodular, but that has this desired property.  We will use an example from \cite{FGL-13}, in which agent 1 is XOS and agent 2 is unit-demand. 
A key property of this example is that the XOS agent always wants all items and the unit-demand agent always wants exactly one item.  Our approach will be to massage these valuations into a pair of submodular functions, while retaining the lower bounds.

Call the original two valuations $v_1$ and $v_2$.  Our first idea is to make the functions submodular by adding a very large, strictly concave function $h$ to each of $v_1$ and $v_2$.  This will make the contribution of $v_1$ and $v_2$ to the social welfare negligible by comparison, but recall that a stable outcome exists only if the integrality gap is \emph{exactly} one, so it is enough for even a tiny component of the welfare to be improvable via fractional allocation.

Unfortunately, this modification ruins the properties of the original example.  Specifically, because of the concavity of $h$, the optimal assignment divides all items equally between the players, and this does not align with the original example. We therefore apply an additional trick: we divide the items into $\sqrt{m}$ equal-sized buckets, and implement the original example on each bucket separately.  


We now begin with the formal proof, which requires some preparation.
Define $h(z) = z + \sum_{i=1}^{z}1/i$.  Function $h$ is symmetric, submodular, and satisfies the following property.

\begin{claim}
\label{claim:submodular}
For every set function $f(\cdot)$, there is a sufficiently small $\varepsilon>0$ such that the function $v(S) = h(|S|) + \epsilon \cdot f(S)$ is submodular for any $\varepsilon\ge\epsilon\ge 0$.
\end{claim}


\begin{proof}
In order to ensure that $v(\cdot)$ is submodular we show below that marginal values are decreasing.
Let $S_1\subsetneq S_2$. We need to ensure that
\[
v(S_1\cup\{i\})-v(S_1)\ge v(S_2\cup\{i\})-v(S_2).
\]
This inequality for a fixed $\epsilon$ is equivalent to the following
\[
\frac{1}{|S_1|+1}-\frac{1}{|S_2|+1}\ge\epsilon\cdot\left(f(S_2\cup\{i\})-f(S_2)- f(S_1\cup\{i\}) +f(S_1)\right)
\]
We observe that the left hand side of the last inequality has a lower bound of $\frac{1}{(n+1)(n+2)}$ and the right hand side has an upper bound of $2
\epsilon\cdot\max_{S\subseteq [n]}f(S)$. Therefore, the inequality holds true for any $\epsilon<\varepsilon=\frac{1}{2(n+1)(n+2)\max_{S\subseteq [n]}f(S)}$.
\qed
\end{proof}

We are now ready to construct our valuations.
Consider an instance with 2 buyers and $m=k^2$ items for some integer $k>0$.
Suppose that the $k^2$ items are partitioned into $k$ buckets, $B_1, \ldots, B_k$, where $|B_i|=k$ for every $i=1, \ldots, k$.

The valuation functions of the buyers are as follows.
Define $XOS(S)=|S|$ for every $S$ such that $|S|> 1$, and $XOS(S)=2$ if $|S|=1$.  
The valuation function of the first buyer is given by
\begin{align*}
v_1(S) & = h(|S|) + \varepsilon \cdot f_1(S),\quad \text{ where } 
f_1(S) = \max_{j=1,\ldots,k} XOS(S\cap B_j).
\end{align*}
We now turn to the second buyer.
Let 
$Unit(S)=1-\frac{1}{k}$  for every $S$ such that $|S|\geq1$ (and $0$ if $S=\emptyset$).  Note $Unit(S)$ is a symmetric unit-demand function.
The valuation function of the second buyer is given by
\begin{align*}
v_2(S) & = h(|S|) + \varepsilon \cdot f_2(S),\quad \text{ where }
f_2(S) =\sum_{j=1}^{k} Unit(S\cap B_j).
\end{align*}
By Claim~\ref{claim:submodular}, we can take sufficiently small $\varepsilon$ such that both valuation functions $v_1(\cdot)$ and $v_2(\cdot)$ are submodular.

Recall that our solution concept allows for some items to be unsold. Our next argument is similar in spirit to the argument for XOS valuations given in \cite{FGL-13} (see Appendix~\ref{app:xos-sim}). Let $K$ denote the set of items sold. The maximum of $h(S)+h(K\setminus S)$ is attained when $|S|$ is as close to $|K|/2$ as possible. Therefore, for a sufficiently small $\varepsilon$ any integral allocation that maximizes $v_1(S)+v_2(K\setminus S)$ must divide items equally among two buyers (up to a single item). Among these,
the optimal integral allocation is one that maximizes $f_1(S)+f_2(K\setminus S)$.
We can now show that if $K$ is not too small, then a stable outcome cannot exist.

\begin{claim}
\label{cl:submodular-frac}
If $|K|\ge 4k$, then there exists a fractional allocation that does strictly better than the optimal integral allocation.
\end{claim}

\begin{proof}
Suppose that $|K|\geq 4k$. Without loss of generality let $B_1$ be the bucket with the largest number of sold items, say $t$. We note that the number of items $t$ in $B_1$ is at least $4$. 


One can easily verify that the allocation that maximizes $f_1(S)+f_2(K\setminus S)$ is one that assigns all items in $B_1$ to the first buyer, and a single item from every remaining non-empty bucket to the second buyer. Indeed, this allocation maximizes the value of $f_1(\cdot)$ and almost matches the maximum of $f_2(\cdot)$, and these two cannot be maximized simultaneously.
Let $S_1$ and $S_2$ be sets of items assigned respectively to the first and second buyer in this allocation. We note that $|S_1|\le k$ and $|S_2|\le k-1$. Therefore, the remaining at least $4k-(2k-1)=2k+1$ unallocated items can be divided among the two buyers into two sets $J_1$ and $J_2$ such that buyers receive the same (up to one) number of items, as required. Thus in the optimal integral allocation the first buyer receives $S_1\cup J_1$ and the second buyer receives $S_2\cup J_2$.

By the characterization given in \cite{BM97} (see also Section~\ref{sec:config-lp}), this allocation admits a stable pricing if and only if this is also an optimal fractional allocation. We next construct a fractional allocation with a higher social welfare.

We observe that $|J_2|>k$, so we can choose $T\subset J_2$ such that $|T|=|S_1|=t.$ We recall that $t\ge 4.$ Let $\pi: S_1\to T$ be a bijection between items in $S_1$ and $T$. We consider the following fractional solution $\{y_{i,S}\}$:

\begin{eqnarray*}
&y_{1,S_1\cup J_1} &= \frac{t-2}{t-1}\\
&y_{1,\{j\}\cup J_1\cup T\setminus\{\pi(j)\}} &= \frac{1}{t(t-1)}, \quad\text{ for each } j\in S_1\\
&y_{2,\{j\}\cup S_2\cup J_2\setminus\{\pi(j)\}} &= \frac{1}{t}, \quad\text{ for each } j\in S_1
\end{eqnarray*}
Let $\ell$ be the number of nonempty buckets in $K$.
One can easily verify that this is a feasible solution, and the welfare obtained by $\{y_{i,S}\}$ is given by
\begin{align*}
SW(y) &= h(|S_1\cup J_1|)+h(|S_2\cup J_2|)+ \frac{t-2}{t-1}\cdot f_1(S_1\cup J_1)  \\
&+\sum_{j\in S_1}\frac{1}{t(t-1)}\cdot f_1(\{j\}\cup J_1\cup T\setminus\{\pi(j)\}) \\
&+\sum_{j\in S_1}\frac{1}{t}\cdot f_2(\{j\}\cup S_2\cup J_2\setminus\{\pi(j)\})\\
&\ge h(|S_1\cup J_1|)+h(|S_2\cup J_2|)+ \frac{t-2}{t-1}\cdot t  \\
&+|S_1|\cdot\frac{1}{t(t-1)}\cdot 2 +\frac{1}{t}\cdot\sum_{j\in S_1} f_2(j\cup S_2)\\
&= h(|S_1\cup J_1|)+h(|S_2\cup J_2|)+ \frac{t-2}{t-1}\cdot t + \frac{2}{t-1}+ \ell\cdot\left(1-\frac{1}{k}\right).
\end{align*}
The welfare in the optimal integral allocation $x$ is given by
\[
SW(x) = h(|S_1\cup J_1|)+h(|S_2\cup J_2|) + t + (\ell-1)\cdot\left(1-\frac{1}{k}\right).
\]

It follows that
$
SW(y)-SW(x)\ge 1-\frac{1}{k} - \frac{t-2}{t-1}>0.
$\qed
\end{proof}

We conclude that no allocation of at least $4k$ items admits a stable outcome. Now since $k = \sqrt{m}$ and $v_1(S)=\Theta(|S|)$, $v_2(S)=\Theta(|S|)$, this gives a gap of $\Omega(\sqrt{m})$ in the social welfare, completing the proof of Theorem \ref{thm:submod.sim}.\qed
\end{proof}




\section{Single-Minded Valuations}
\label{sec:singleminded}
We now consider buyers with single-minded valuations.
A single minded buyer is interested in a single set and derives no utility from any strict subset of it.
Single minded buyers exhibit strong complementarity between items.


We first present an instance such that there does not exist a stable outcome that obtains better than $\Omega(\sqrt{m})$ approximation to the optimal social welfare.
Readers who are familiar with the NP-hardness of welfare approximation within a factor $\Omega(\sqrt{m})$ for single minded bidders (e.g., \cite{LOS99}) might assume that our $\Omega(\sqrt{m})$ lower bound follows similarly.
We reiterate here that our lower bound is purely existential and is completely independent of computational considerations, and is therefore unrelated to the computational result.

\begin{example}
Consider a combinatorial auction with $n$ single-minded buyers and $m=\frac{n(n-1)}{2}$ items. Each item is represented by a pair of integers $(a,b)$
such that $M=\{(a,b)\mid a+b\le n,\quad a,b\in[n]\}$. Each buyer $i$ for $i\in[n-1]$ desires a set $S_i^{*}=\{(a,b)\mid a=i\text{ or }b=n-i,\quad (a,b)\in M \}$ of size $n-1$, for a value of $n+1$. Note that $|S_{i}^{*} \cap S_{k}^{*}| = 1$ for all $i \neq k$, and every object $a \in M$ is contained in at most $2$ sets. Buyer $n$ desires the set of all items, $M$, for a value of $m$.

The optimal outcome allocates all objects to buyer $n$, for a total value of $m$. Suppose this optimal outcome can be supported by item prices $\prices$.
Then $\sum_{a \in M} \pricei[a]\le m$, but a simple counting argument below demonstrates that \newline$\sum_{a \in S_{\ell}^{*}} \pricei[a] \le n < n+1$ for some
$\ell$.
\[
n(n-1)=2\vali[n](M)\ge 2\sum_{a \in M} \pricei[a]\ge\sum_{i=1}^{n-1}\sum_{a\in S_{i}^{*}}\pricei[a].
\]
Buyer $\ell$ would demand set $S_{\ell}^{*}$ at these prices, a contradiction. Thus, at any pricing equilibrium, some set $S_{\ell}^{*}$ must be allocated to a buyer $\ell$; but this implies that no other set $S_{k}^{*}$ can be allocated, as $S_{k}^{*} \cap S_{i}^{*} \neq \emptyset$ for all $k \neq i$.  Thus the total social welfare at any equilibrium outcome with item prices is at most $n+1=O(\sqrt{m})$.
\label{ex:single-minded}\qed
\end{example}

Finally, we note that the last example is essentially the worst possible case.
That is, for every instance of single-minded buyers, there exists a stable outcome that achieves at least an $O(\sqrt{m})$ approximation to the optimal social welfare.




The latter statement is implied by \cite{Cheung2008}, i.e., the agrument goes through a careful analysis of the configuration LP along with its dual program.
In particular, one can show how to construct a price vector for which (i) no buyer can get strictly positive utility, and (ii)
there exists a fractionally optimal solution in which all buyers obtain zero utility.
Then, a standard greedy algorithm applied to the buyers in the support of the fractional
solution will yield an $O(\sqrt{m})$ approximation.
Since it is not possible for any buyer to obtain strictly positive utility, this outcome is necessarily stable.



\section{Conclusions and Open Problems}
\label{sec:conclusions}
This paper studies the problem of resolving combinatorial markets using item prices and utility-maximizing allocations, with the objective of maximizing economic efficiency.
We approach the problem through the lens of approximation, and ask how much social welfare can be obtained using item prices.
We establish negative results, showing that no stable pricing can provide reasonable approximation to social welfare in several cases of interest.
These negative results hold even for simple complement-free valuations, and even with only two buyers.

\paragraph{Open problems}

Our model and results leave a number of directions for future research.
%

First, we show that even with only two submodular buyers, there is a lower bound of $\Omega(\sqrt{m})$ on the approximation ratio of social welfare.
On the positive side, for gross-substitute buyers there always exists a competitive equilibrium, i.e., efficient stable pricing. 
It is left open whether there are subclasses of submodular valuations (such as budget additive or coverage valuations) that admit constant approximation to optimal welfare.
For the case of two consumers with budget additive valuations we answer this question in the affirmative.
However, for larger instances of budget additive valuations or other valuation classes this question is left open.

Second, throughout the paper we assume that items are indivisible and heterogeneous. It would be an interesting research direction to relax partially these assumption. For example, one could assume that every item in the market has a few identical copies and that every buyer does not want more than one copy of each item. It would be interesting to see how the efficiency of item pricing depends on the minimal number of item copies. Given our mostly negative results for valuations with complements, we would like to understand under what relaxations one can obtain positive results (e.g., constant approximation of the optimal social welfare) for single-minded buyers.

\section*{Acknowledgments}
The work of Michal Feldman was partially supported by the European Research Council under the European Union's Seventh Framework Programme (FP7/2007-2013) / ERC grant agreement number 337122.

\bibliographystyle{plain}
\bibliography{item-prices}

\medskip

\appendix
\section*{APPENDIX}
\setcounter{section}{0}

\section{A lower bound for XOS valuations}
\label{app:xos-sim}

For completeness, we present an example from \cite{FGL-13}, demonstrating a linear gap in welfare for XOS valuations.

\begin{theorem}[\cite{FGL-13}]
There exists an instance with two buyers with XOS valuations, such that there does not exist a stable outcome that obtains better than $\Omega(m)$ ($\ge m/3$) approximation to the optimal social welfare.
\end{theorem}
\begin{proof}
Consider an auction with $m$ items and two buyers with the following symmetric XOS valuations.
Buyer 1 is unit-demand and values every subset at $1/2-\delta$, for a sufficiently small $\delta$ (that will be determined soon).
Buyer 2 values any subset of size $k$ at $\max(1,k/2)$; it is easy to verify that this is an XOS valuation. The socially optimal outcome allocates all objects to the second buyer for a total value of $m/2$ when $m>2$, and allocates one item to each of the buyers when $m=2$ for a total value of $3/2-\delta$. We claim that there is no stable pricing that sells more than two items. For every $m > 2$, an optimal integral solution obtains a value of $m/2$ (by giving all the items to the XOS buyer). By the characterization given in \cite{BM97} (see also Section~\ref{sec:config-lp}), this allocation admits a stable pricing if and only if $m/2$ is the optimal fractional solution of the corresponding configuration LP. We will now show a fractional solution that obtains value greater than $m/2$ for every $\delta < \frac{1}{2(m-1)}$.
Consider the fractional solution in which the allocation of the first (unit demand) buyer is given by $x_{1,\{j\}}=1/m$ for every $j \in [m]$, and the allocation of the second (XOS) buyer is given by $x_{2,\{j\}}=\frac{1}{m(m-1)}$ for every $j \in [m]$, and $x_{2,[m]}=\frac{m-2}{m-1}$.
One can easily verify that this is a feasible solution, and the welfare obtained by $\{x_{i,S}\}$ is given by $SW(x)=\frac{m}{2}+\frac{1}{2(m-1)}-\delta$,
which is greater than $\frac{m}{2}$ for every $\delta < \frac{1}{2(m-1)}$, as required. We conclude that a stable outcome can have at most two allocated object, and thus the highest welfare that can be obtained in a stable allocation is $3/2-\delta$, resulting in a linear gap of $m/3$.
\end{proof}



\section{Two budget additive valuations: an upper bound}
\label{app:budget-additive}

It has been shown in the last two sections that for submodular and XOS valuations, one cannot hope for a constant approximation, even for instances with only two buyers. One may wonder whether this is an unavoidable property of any valuation beyond gross substitutes.
In this section we show that this is still too early to jump into this conclusion.
In particular, for the case of two budget-additive buyers, constant approximation can be achieved.

A budget-additive valuation of buyer $i$ is specified by budget $B$ and item values $\val_{ij}$ for every $j \in M$.
The value of some set $S$ is then $\vali(S) = \min\{B, \sum_{j \in S} v_{ij}\}$. Without loss of generality we may assume that for each item $j$ neither $\vali[1j]\le B_1$, nor $\vali[2j]\le B_2$.

\begin{theorem}
For every instance of two budget-additive buyers, there exists a stable outcome, which is $4$-approximation to the optimal social welfare.
\end{theorem}

\begin{proof}
Let $B_1$ and $B_2$ denote the budgets of buyers 1 and 2, respectively, and suppose without loss of generality that $B_1 \geq B_2$.
We distinguish between two cases.

\vspace{0.1in}\noindent {\bf Case 1:} $\sum_j \val_{1j} \geq B_1$ (i.e., the value of buyer 1 for the grand bundle exhausts her budget).
In this case, set a price of $\price_j = \val_{1j}$ for every item $j$.
Thus, every bundle $S$ (including $S=\emptyset$) such that $\sum_{j \in S} \val_{1j} \geq B_1$ is in buyer 1's demand set (and gives her utility 0). Since $\vali[1j]\le B_1$ for each $j$ and $\sum_j \val_{1j} \geq B_1$ we can find a set $S_1$, such that $B_1\ge\sum_{j \in S_1} \val_{1j} \ge B_1/2$. Any subset of $S_1$ is a demand set of buyer $1$. Let $D_2$ be a bundle in buyer 2's demand set. We let $D_2$ be the allocation of buyer $2$. We note that
\[
\vali[2](D_2)\ge\sum_{j\in D_2}\pricei[j]=\sum_{j\in D_2}\vali[1j].
\]

We let $D_1=S_1\setminus D_2$ be the allocation of buyer $1$. We can estimate the social welfare of this allocation as follows.

\begin{align*}
\SW=\vali[1](D_1)+\vali[2](D_2)&=\sum_{j\in S_1\setminus D_2}\vali[1j]+\vali[2](D_2)\\
&\ge
\sum_{j\in S_1\setminus D_2}\vali[1j]+\sum_{j\in D_2}\vali[1j]\ge\sum_{j\in S_1}\vali[1j]\ge\frac{B_1}{2}.
\end{align*}

The optimal social welfare is at most $B_1+B_2$, which cannot exceed our $\SW$ by more than a factor of $4$.

%

\vspace{0.1in}\noindent {\bf Case 2:} $\sum_j \val_{1j} < B_1$.
Let $S_1 = \{j : \val_{1j} \geq \val_{2j}\}$, and
$S_2 = \{j : \val_{2j} > \val_{1j}\}$.
Consider the pricing where $\price_j = \val_{2j}$ for every item $j \in S_1$ and $\price_j = \val_{1j}$ for every item $j \in S_2$.
Since buyer 2 is indifferent about taking any item from $S_1$, there exists a set in buyer 2's demand set that is contained in $S_2$, we call it $D_2$.
We let buyer 2 to be allocated the items in $D_2$ and buyer 1 be allocated the items in $S_1 \cup (S_2 \setminus D_2)$. The social welfare $\SW$ for this allocation is
\[
\SW=\sum_{j\in S_1}\vali[1j]+\sum_{j\in S_2\setminus D_2}\vali[1j]+\min\InParentheses{\sum_{j\in D_2}\vali[2j],B_2}.
\]

We assume that $D_2\neq\emptyset,$ as otherwise we would have the optimal social welfare.
If $\val_2(D_2) + \val_1(S_2 \setminus D_2) < B_2$, then $\utili[2](S_2)>\utili[2](D_2)$ and we arrive at a contradiction with the fact that $D_2$ was a demand set of buyer $2$. Indeed,
\begin{align*}
\utili[2](D_2)&=\val_2(D_2) -\sum_{j\in D_2}\vali[1j] = \val_2(D_2) + \val_1(S_2 \setminus D_2)-\sum_{j\in S_2}\vali[1j]\\
& <  \min\InParentheses{\sum_{j\in S_2}\vali[2j],B_2}-\sum_{j\in S_2}\vali[1j] \quad = \quad \vali[2](S_2)-\sum\limits_{j\in S_2}\pricei[j]\quad=\quad\utili[2](S_2)
\end{align*}



Therefore, $\val_2(D_2) + \val_1(S_2 \setminus D_2) \geq B_2$, then it also holds that
\begin{equation}
\SW=\val_1(D_1) + \val_2(D_2) + \val_1(S_2 \setminus D_2) \geq B_2.
\label{eq:1}
\end{equation}
In addition to that, since $\vali[2](D_2)\ge\sum_{j\in D_2}\pricei[j]=\sum_{j\in D_2}\vali[1j]$, it follows that
\begin{equation}
\SW=\val_1(D_1) + \val_2(D_2) + \val_1(S_2 \setminus D_2) \geq \sum_{j \in M}\val_{1j}.
\label{eq:2}
\end{equation}
Consider the sum of Equations (\ref{eq:1}) and (\ref{eq:2}).
The left hand side is twice the social welfare obtained by the suggested allocation,
and the right hand side is $B_2 + \sum_{j \in M}\val_{1j}$, which is clearly an upper bound on OPT.
It follows that we get at least $1/2$ of the optimal social welfare.
\end{proof}


\appendix

\end{document}